\documentclass[aps,prl,reprint]{revtex4-2}
\usepackage{tikz}
\usepackage{float}
\usepackage{color}
\usepackage{xcolor}
\definecolor{blue}{rgb}{0,0.5,1}
\usepackage{braket}
\usepackage{amsmath}
\usepackage{amsthm}
\usepackage{amsfonts}

\usepackage{hyperref}

\usepackage{amssymb}
\usepackage{graphicx} % Required for inserting images

\DeclareMathOperator{\Tr}{Tr}

\newtheorem{theorem}{Theorem}

\theoremstyle{remark} % Uses italicized text, no bold header
\newtheorem*{remark}{Remark}
\theoremstyle{definition} % This gives upright text (not italic) and adds space.
\newtheorem{definition}{Definition}

\begin{document}

\title{Necessary and sufficient conditions for the $N$-representability of functionals of the 1-electron reduced density matrix}
\author{Jannis Erhard}%
\affiliation{Department of Chemistry \& Chemical Biology, McMaster University,\\ 1280 Main St.\ West, Hamilton, Ontario, L8S 4M1, Canada}
\author{Paul W.\ Ayers}%
\email{ayers@mcmaster.ca}
\affiliation{Department of Chemistry \& Chemical Biology, McMaster University,\\ 1280 Main St.\ West, Hamilton, Ontario, L8S 4M1, Canada}\date{\today}

\begin{abstract}
We establish necessary and sufficient conditions for the $N$-representability of the universal one-electron reduced density matrix functional. Functionals satisfying these conditions are guaranteed to yield variational upper bounds on the true energy in one-electron reduced density matrix functional theory, regardless of the strength of the interparticle repulsion. Conversely, any functional violating these conditions will necessarily underestimate the true energy for certain systems. These exact constraints 
impose a stringent restriction on density matrix functional approximations, as many existing functionals---including the Hartree-Fock functional---appear to violate them. {This mathematical formalism, therefore, can guide the development of new approximate functionals and numerical algorithms.}
\end{abstract}
\maketitle

\section{Introduction}
One-electron reduced density matrix functional theory (1DMFT) offers a promising alternative to traditional approaches like DFT for tackling strongly-correlated quantum systems \cite{Piris_2024}. Unlike methods based on the wavefunction or electron density, 1DMFT treats the 1-electron reduced density matrix (1DM) as the fundamental variable \cite{Loewdin_I, Loewdin_II,mezeyManyelectronDensitiesReduced2000}. The theoretical foundation of 1DMFT was established by Gilbert \cite{GilbertExtension}, who showed that the ground-state energy can be expressed as a universal functional of the 1DM, while Levy and Valone later provided a rigorous formulation based on the constrained search\cite{LevyValoneLevy, LevyValoneValone}. Other, more complicated, formulations based on constrained search and Legendre transform have also been proposed.\cite{ayersGeneralizedDensityfunctionalTheory2005,PaulGolden,feliciangeliNoncommutativeEntropicOptimal2023}

 While density-functional theory (DFT) benefits from a well-developed formal framework, 1DMFT currently lacks an equally rigorous foundation \cite{LiebOnDFT,Cohen2008,SavinAC,Savin_2003, Goerling_1996,tealeAccurateCalculationModeling2010,
tealeCalculationAdiabaticconnectionCurves2009,ayersAxiomaticFormulationsHohenbergKohn2006,lewinUniversalFunctionalsDensity2023,
tealeLiebVariationPrinciple2022}, though recent work by Fredheim and Kvaal has started to bridge that gap \cite{fredheimReducedDensityMatrix2025}. The most accurate and versatile functionals in DFT rely on deep insights into the theoretical framework of DFT to derive properties of the exact functional, which are then used to constraint the construction of approximate functionals\cite{Scuseria_2008,PBE, Becke, SCAN, scaledSigma, PSA,appl_rep_cond_1, appl_rep_cond_2}. The absence of a similarly systematic approach to 1DMFT makes functional construction more challenging, as approximations often rely on empirical corrections or physically motivated ans\"atze \cite{PirisConcurrent}.
Nevertheless, significant progress has been made in developing approximate functionals, including the Müller functional \cite{MULLER_1984}, the power functional \cite{Sharma2008}, the BBC functionals \cite{Baerends_2005}, the family of Piris natural orbital functionals \cite{PNOF5,PNOF6,PNOF7,GNOF}, and, most recently, machine-learned functionals \cite{MachineLearned1RDMFT}.  

Given the exact 1DM functional (or a sufficiently accurate approximation thereto), the ground-state energy for an $N$-electron system whose 1-electron operator (containing information about the kinetic energy, atomic locations, and external fields), $h$, can be deduced from the variational principle
\cite{GilbertExtension}:

\begin{eqnarray}
\label{eq:GilbertHK_1}
% E_{g.s.}[h_{v}; N] = \min_{\{\gamma | \gamma^{\dagger} = \gamma, Tr[\gamma] = N, 0 \preceq \gamma \}} \{  \Tr[h_{v} \gamma] + V_{ee}[\gamma] \} \min_{\left \{{\gamma \Big | \substack{\gamma^{\dagger} = \gamma, 
% 0 \preceq \gamma ,  \\ \text{Tr}[\gamma ]= N
% }}\right \}}
E_{g.s.}[h; N] =\min_{\left \{{ \gamma \Big | \substack{\gamma^{\dagger} = \gamma\\ 
0 \preceq \gamma \preceq 1  \\ \text{Tr}[\gamma ]= N
}} \right\}}  \{  \Tr[ h \gamma] + V_{ee}[\gamma] \} \nonumber \\
\end{eqnarray}
The restrictions on the variational domain ensure that the 1DM, denoted $\gamma$, is ensemble-$N$-representable.\cite{ColemanVariationandNrep} A 1DM is ensemble $N$-representable if and only if  there exists an $N$-electron mixed state, 
\begin{equation}
\Gamma = \sum_i p_i \ket{\Psi_i}\bra{\Psi_i}
\end{equation}
that has the specified 1DM,
\begin{align}
\label{eq:define_Nrep_1DM}
\gamma[\{{p_i, \Psi_i}\}] = \sum_i p_i \sum_{qr} \langle \Psi_i |  q^{\dagger}  r | \Psi_i \rangle | q \rangle \langle r |.
\end{align}
The $p_i$ are the ensemble weights for the $N$-electron wavefunctions $\ket{\Psi_i}$. {Here $q^{\dagger}$ and  $r$ denote the second-quantized operators that create and annihilate elements of an orthonormal basis for the one-electron Hilbert space, denoted $| q\rangle$ and $|r \rangle $ .}

In Eq. (\ref{eq:GilbertHK_1}), $V_{ee}[\gamma] $ is the electron-electron repulsion energy, expressed as a universal functional of the 1DM. The exact functional can be defined, for example, by the Levy-Valone constrained search,\cite{LevyValoneLevy, LevyValoneValone}
\begin{equation}
\label{eq:constrained_search_LV}
V_{ee}^{LV}[\gamma] = \min_{\Set{\Gamma|\Gamma \rightarrow \gamma}} \text{Tr}[V_{ee}\Gamma]
\end{equation}
where the $\Gamma \rightarrow \gamma$ is a concise notation indicating that the domain of the minimization is constrained to $N$-electron density matrices that satisfy Eq. (\ref{eq:define_Nrep_1DM}).

Known properties of the exact functional can be used to constrain what values of $V_{ee}$ are acceptable for a given $\gamma$. In this paper, we are concerned with the $N$-representability of $V_{ee}[{\gamma}]$, in analogy to discussions of $N$-representability of density functionals\cite{mori-sanchezManyelectronSelfinteractionError2006,
ludenaFunctionalNrepresentability2matrix2013,
ludenaNrepresentabilityUniversalityFrho2009,
kryachkoFormulationNRepresentableUpsilonRepresentable1991,
ludenaHohenbergKohnShamVersionDFT2004,AyersNrepfDFTFs,cuevas-saavedraSymmetricNonLocal2012}. {While every $N$-representable density functional, $V_{ee}[\rho]$, is formally an $N$-representable  density-matrix functional, $V_{ee}[\rho[\gamma]]$ (recall $\rho(\mathbf{r}) = \gamma(\mathbf{r},\mathbf{r})$), our aim is to develop a general mathematical framework for deducing N-representability conditions on $\gamma$.
}
{We added a formal definition of what it means for a functional to be $N$-representable near the top of page 2. To wit:
\begin{definition}[N representable functional]
A functional $V_{ee}[\gamma]: \gamma \to V_{ee}$ is \emph{ensemble-$N$-representable} if there exists at least one $N$-electron state, $\Gamma$, such that simultaneously equation (\ref{eq:define_Nrep_1DM}) yields $\gamma$ and, simultaneously,
\begin{align}
\label{eq:define_Nrep_Functional}
V_{ee}[\{{p_i, \Psi_i}\}] = \sum_i p_i \langle \Psi_i | V_{ee} | \Psi_i \rangle,
\end{align}
\end{definition}}

Approximate functionals are typically non-$N$-representable.\cite{Donnelly_1979, HarrimanAndHerbert,cioslowskiVariationalDensityMatrix2002} To our knowledge, the only $V_{ee}[{\gamma}]$ functionals that have been 
designed with $N$-representability conditions in mind are the Piris functionals,\cite{Piris_2006,PirisBenchmark,Piris_2024,Piris_2024b} wherein the 
two-electron reduced density matrix, $\Gamma_2$, is written a functional of ${\gamma}$, so that
\begin{equation}
V_{ee}[\gamma] = \text{Tr}[V_{ee} \Gamma_2[{\gamma}]]
\end{equation}
Then enforcing \textit{approximate} $N$-representability conditions on $\Gamma_2$ ensures \textit{approximate} $N$-representability of $V_{ee}[{\gamma}]$. {We note that alternative approaches have recently been explored, including frameworks that incorporate information from both the one- and two-particle density matrices, thereby which can make it easier to formulate (approximately) $N$-representable functionals \cite{gebauerWellscalingNaturalOrbital2016,ayersDensityfunctionalTheoryAdditional2009b,Maziotti2023,senjeanReducedDensityMatrix2022}.}
In this paper, we present the \textit{exact} conditions that $V_{ee}[{\gamma}]$ must satisfy to be ensemble-$N$-representable.
%
% to the complexity of enforcing $N$-representability conditions—the requirement that the 1DMF corresponds to a physically valid N-electron wavefunction \cite{PirisBenchmark}. This constraint, formalized by Coleman in terms of ensemble $N$-representability (the 1DM must derive from a mixed-state density operator), introduces additional challenges in functional construction, as violations can lead to unphysical results \cite{ColemanVariationandNrep, Donnelly_1979, HarrimanAndHerbert}. Thus, while 1DMFT offers a promising alternative to DFT for capturing strong correlation effects, its practical success hinges on further advances in both functional approximations and the rigorous incorporation of $N$-representability \cite{Piris_2024b, Julia2023,Schilling2018,Schilling2019,liebert2025refiningensemblenrepresentabilityonebody}.
%
To this end, in section \ref{sec:sets}, we define abstract spaces and their properties so that we can confidently leverage theorems from functional analysis. Then, in Section \ref{sec:theorems}, we establish necessary and sufficient conditions for $V_{ee}[{\gamma}]$ to be $N$-representable and a bivariational principle for the ground-state energy. In section \ref{sec:numerics} we provide a numerical demonstration of the $N$-representability conditions and explicitly demonstrate that the Hartree-Fock $V_{ee}$ functional is \textit{not} $N$-representable.

\section{Sets and Spaces}
\label{sec:sets}
Let $\mathcal{B}_{HS}(\mathcal{H})$ denote the space of bounded linear operators on the  one-electron Hilbert space, $\mathcal{H}$. The ensemble-$N$-representable 1DMs are a closed, convex subset of $\mathcal{B}_{HS}(\mathcal{H})$.\cite{ColemanVariationandNrep} Similarly, the real numbers, $\mathbb{R}$, define a one-dimensional Hilbert space. For the repulsive Coulomb interaction, ensemble-$N$-representable $V_{ee}$ are a closed, convex subset of $\mathbb{R}$, namely the nonnegative real numbers, $\mathbb{R}_{\geq 0}$. However, our analysis does not use properties of the Coulomb interaction, and is easily extended to near-arbitrary interparticle repulsions.  

One-electron reduced density matrix functionals  are sets of ordered pairs of functional value and 1DM $(V_{ee}, {\gamma})$, hence they are embedded in the Cartesian product space of the ambient space of both previously described sets.
\begin{align}
    \mathcal{V} = \left\{ \left(W, \eta \right) \Bigm| W \in \mathbb{R},\; \eta \in \mathcal{B}_{HS}(\mathcal{H}) \right\}
\end{align}
The dual space of a cartesian product of finitely many vector spaces is the direct sum of the dual spaces of the individual primal spaces. By virtue of Riesz Representation Theorem, the dual-space of a Hilbert space is isometrically isomorphic to the Hilbert space itself \cite{FuncAna}. 
% \begin{eqnarray}
%   \mathcal{V}^{\ast} \cong   (\mathbb{R} \oplus \mathcal{B}_{HS}(\mathcal{H}))
% \end{eqnarray}

Let $\mathcal{F}_{N} \subset \mathcal{V}$ denote the set of ensemble-$N$-representable $V_{ee}[{\gamma}]$. 
An element of the ambient space, $(W, \eta)$, is ensemble-$N$-representable if $\eta$ and $W$ simultaneously satisfy Eqs. (\ref{eq:define_Nrep_1DM}) and (\ref{eq:define_Nrep_Functional}), respectively. I.e., 
\begin{widetext}
\begin{eqnarray}\label{eq:FN_1DMF}
\mathcal{F}_{N}
  = \Bigl\{\, (W,\eta)\;\Big|\;
      \exists\,\{p_i,\Psi_i\}_{i=1}^{N}\!
      :\;
      \eta = \gamma\!\bigl[\{p_i,\Psi_i\}\bigr],\;
      W   = V_{ee}\!\bigl[\{p_i,\Psi_i\}\bigr]
    \Bigr\}.
\end{eqnarray}
\end{widetext}
$\mathcal{F}_{N}$ is closed and convex since it is defined by a linear map from the (closed and convex) set of $N$-electron density matrices, $\Gamma$.
 
\section{Theorems}
\label{sec:theorems}
\begin{theorem}
A 1DM functional, $(V_{ee}, {\gamma})$, is ensemble-N-representable if and only if, for all one-electron operators $h$, and interaction strengths, $\lambda \in (-\infty,\infty)$,  

\begin{equation}
\label{eq:condition_theorem_I}
\operatorname{Tr}[h \gamma] + \lambda V_{ee}[\gamma] \geq E_{g.s.}^\lambda[N, h ],
\end{equation}
where $E_{g.s.}^\lambda[N,  h ]$ is the $N$-electron ground-state energy for the Hamiltonian $H = h + \lambda V_{ee}$ and $N=\Tr[\gamma]$.

% This means that for any external potential $v$ and any one-electron reduced density matrix $\gamma$, the the left-hand side must be greater than the ground-state $E_{g.s.}$ energy of that system. 
\end{theorem}

\begin{proof}
% All functionals that are $N$-representable fulfill the theorem, though their  might be functionals that fulfill the theorem that are non $N$-representable. 

Consider the definition of the ground-state energy as a variational minimization problem over ensembles:

\begin{eqnarray}
    E_{g.s.}^\lambda[N, h ] &=& \min_{p_i, \Psi_i}  \sum_{i} p_i \left \langle \Psi_i \Bigg | h + \lambda V_{ee} \Bigg | \Psi _i \right \rangle  \nonumber \\
    &=& \min_{p_i, \Psi_i} \left( \text{Tr}[ h  \gamma[\{p_i, \Psi_i\}]+\lambda V_{ee}[\{p_i, \Psi_i\}]  \right) \nonumber \\
\end{eqnarray}
% where the definition of $N$-representable functional \ref{eq:define_Nrep_Functional} and density matrix \ref{eq:define_Nrep_1DM} was inserted.
Referring to the definition of $\mathcal{F}_N$, the second line can be rewritten as the variational minimization over the set of $N$-representable functionals,
\begin{equation}\label{eq:variational_Fn}
E_{g.s.}^\lambda[N,  h ] = \min_{(V_{ee}, \gamma)\in \mathcal{F}_N} \text{Tr}[h  \gamma] + \lambda V_{ee}[ \gamma]
\end{equation}
Thus Eq. (\ref{eq:condition_theorem_I}) is necessary for $(V_{ee}, {\gamma}) \in \mathcal{F}_N$.

% The pair $(V_{ee}[\{p_i, \Psi_i\}], \gamma[\{p_i, \Psi_i\}])$ minimizing the expression, is the minimizing $N$-representable pair by construction. A pair lowering the energy would lie outside the set of $N$-representable pairs, and must therefore be non-$N$-representable. Hence, equation \ref{eq:condition_theorem_I} is a necessary condition for $N$-representability.

% % Any pair that fulfills the condition is $N$-representable, although there may be more being  $N$-representable, that don't fulfill the condition. 

To show sufficiency of the condition, choose a trial $(\widetilde{V}_{ee}, \widetilde {\gamma}) \in \mathcal{V}$ that is not $N$-representable. Since the $\mathcal{F}_N$ is a convex subset of a Hilbert space, the hyperplane separation theorem guarantees that there exists an element of the dual space, represented by $(\lambda,g)$, that separates the convex set $\mathcal{F}_{N}$ from the chosen point\cite{FuncAna}: 
\begin{eqnarray}
\label{eq:plane}
    % \Phi[h_{\omega}, k]((V_{ee}[\gamma], \gamma) )= \Tr[h_{w} \gamma] + k V_{ee}[\gamma] 
\text{Tr}[g  \gamma] + \lambda  V_{ee}[\gamma] > \text{Tr}[g \tilde \gamma] + \lambda \tilde{V}_{ee}
\end{eqnarray}
% Note: I use Vee[\gamma] on the left because I mean multiple points and $Vee$ on the right becaus I mean one point.
% Note: Why larger why not smaller, apart from that it would be inconvenient.
for every $(V_{ee}, {\gamma})\in \mathcal{F}_N$. Minimizing the left-hand-side over all $N$-representable functions (recall Eq. (\ref{eq:variational_Fn})) gives  
\begin{equation}
\label{eq:theor_I_c}
  E_{g.s.}^\lambda[N,  g ] >  \Tr[ g {\tilde{\gamma}}] + \lambda \tilde{V}_{ee}. 
\end{equation}
Therefore Eq. (\ref{eq:condition_theorem_I}) is also sufficient for $(V_{ee}, {\gamma}) \in \mathcal{F}_N$.
\end{proof}

\begin{remark}
    In the context of variational 1DMFT calculations, Eq. (\ref{eq:GilbertHK_1}), the theorem indicates that $N$-representable functionals never give an answer below the true ground-state energy. By contrast, a non-$N$-representable functional will always give an answer below the true energy for some system, albeit possibly a system with an attractive interparticle interaction ($\lambda < 0$).
\end{remark}

\begin{remark}
In practice, it suffices to impose Eq. (\ref{eq:condition_theorem_I}) only for $\lambda = \pm 1$,
\begin{equation}
\label{eq:condition_theorem_Ib}
\text{Tr}[h \gamma] \pm V_{ee}[\gamma] \geq E_{g.s.}^{\pm 1}[N, h ],
\end{equation}
This suffices because Eq. (\ref{eq:plane}) can be rewritten as:
\begin{equation}
\text{Tr}[g'\gamma] + \text{sgn}(\lambda)  V_{ee}[\gamma] > \text{Tr}[g' \tilde \gamma] + \text{sgn}(\lambda) \tilde{V}_{ee}
\end{equation}
where $g' = |\lambda|^{-1} g$ and $\text{sgn}(\lambda)$ is the sign of the interaction potential.
\end{remark}

The variational principle in Eq. (\ref{eq:variational_Fn}) is not especially practical because it requires a complete characterization of the set of $N$-representable functionals, $\mathcal{F}_N$. Our second theorem establishes a bivariational principle whereby one can impose only a subset of the necessary conditions for $N$-representability---even just a single condition.
\begin{theorem}
Let $\mathcal{H}^\lambda[N,\tilde{h}]$ denote the half-space of candidate functionals, $(V_{ee},\gamma)$, that satisfy Eq. (\ref{eq:condition_theorem_I}) for a given 1-body potential $\tilde{h}$ and interaction strength $\lambda$. The $N$-electron ground-state energy, \( E_{g.s.}[h; N] \), can be obtained by the bivariational principle:  
\begin{equation}\label{eq:theorem_II}
    E_{g.s.}^\lambda[h;N] = \max_{\tilde{h}\in \mathcal{B}_{HS}(\mathcal{H})} \min_{(V_{ee},\gamma) \in \mathcal{H}^\lambda[N,\tilde{h}]}\left( \Tr[h \gamma] + \lambda V_{ee}[\gamma] \right).
\end{equation}
\end{theorem}

\begin{proof}
    
For any trial one body Hamiltonian  $\tilde{h}$, Eq. \ref{eq:condition_theorem_I} defines a necessary criterion for $N$-representability. By minimizing the energy under this constraint, we obtain a lower bound on the exact ground-state energy:
\begin{equation}\label{eq:theorem_II_prf}
    E_{g.s.}^\lambda[h;N] \geq \min_{(V_{ee}, \gamma) \in \mathcal{H}^\lambda(N,\tilde{h})} \left( \Tr[h \gamma] + \lambda V_{ee}[\gamma] \right). 
\end{equation}
with equality only when $\tilde{h} = h$. As we wish to find the tightest possible lower bound, we maximize over $\tilde{h}$, leading to Eq. (\ref{eq:theorem_II}).
\end{proof}

\begin{remark}
    In practice, the right-hand-side of Eq. (\ref{eq:theorem_II_prf}) is minus infinity, so in practice one wishes to include additional necessary conditions for functional $N$-representability when employing the max-min principle established by Theorem 2.
\end{remark}

\section{Example}\label{sec:numerics}
We consider the one-electron reduced density matrix for 2 electrons in 2 spatial orbitals in its natural orbital (eigen)basis, $\gamma_{ij} = \delta_{ij} n_i$. We consider the case where all spins are paired, so $0 \le n_i \le 2$. For this small system we can explicitly construct a lower bound. Specifically, for a given $\gamma$, the smallest value of $x \in \mathbb{R}_{\ge0}$ for which $(x,\gamma)$ is $N$-representable is
\begin{equation}\label{eq:lower_bound}
\begin{split}
    V_{ee}[\gamma] &= \min_{\Set{x|(x,\gamma)\in \mathcal{F}_N}} x \\
    &=\min_{\{p_i,\Psi_i\} \rightarrow \gamma} \sum_i p_i \langle \Psi_i | V_{ee} | \Psi_i \rangle,
\end{split}
\end{equation}
where the constraint in the second line indicates that Eq. (\ref{eq:define_Nrep_1DM}) is satisfied. The upper bound for $V_{ee}[\gamma]$ is obtained from the lower-bound for an attractive Coulomb interaction,
\begin{equation}\label{eq:upper_bound}
\begin{split}
    V_{ee}^{u.b.}[\gamma] &= \max_{\Set{x|(x,\gamma)\in \mathcal{F}_N}} x \\
    &=-\min_{\{p_i,\Psi_i\} \rightarrow \gamma} \sum_i p_i \langle \Psi_i | -V_{ee} | \Psi_i \rangle,
\end{split}
\end{equation}
{In our numerical work we do not implement the constrained-search functionals \cite{LevyValoneValone} \ref{eq:lower_bound} and \ref{eq:upper_bound} directly, but instead use the equivalent Legendre-transform (dual) formulation.\cite{PaulGolden}} 

{As an example of a non-$N$-representable functional, consider the Hartree-Fock functional, which is obtained by approximating the 2RDM as the wedge product of the 1RDMs, with the matrix elements (in spatial basis)
\begin{eqnarray}\label{eq:2rdm_hf}
    ^2\Gamma^{HF}_{pqrs} = n_r n_s (\delta_{pr} \delta_{qs} - \frac{1}{2} \delta_{qr} \delta_{ps}),
\end{eqnarray}
where $0 \le n_r\le 2$ are the occupation numbers of the spatial natural orbitals. As seen in Figure \ref{fig:numercial_data}, the Hartree-Fock $V_{ee}$ functional is an upper bound to $V_{ee}$ for repulsive interactions, in agreement with Lieb's result \cite{LiebHartreeFock}. However, it is not an upper bound for attractive interactions, so the Hartree-Fock functional is only $N$-representable for density matrices that are sufficiently close to idempotent. This is unsurprising and can be rationalized because the 2RDM in Eq. \ref{eq:2rdm_hf} is non-$N$-representable for $n_r \notin\{0,2\}$ because it violates the trace condition, 
\begin{eqnarray}
    \sum_{rs} {^2\Gamma}^{HF}_{rs,rs} = \left (\sum n_r \right )^2 - \frac{1}{2} \sum n_r^2 \ge N(N-1)
\end{eqnarray} 
For example, at the center of Figure \ref{fig:numercial_data} all the occupation numbers are 1 and the trace is 3. Because the Hartree-Fock functional effectively overestimates the number of electron pairs, it can underestimate the energy for attractive pairing interactions.}

\begin{figure}
    \centering
    \includegraphics[width=0.95\linewidth]{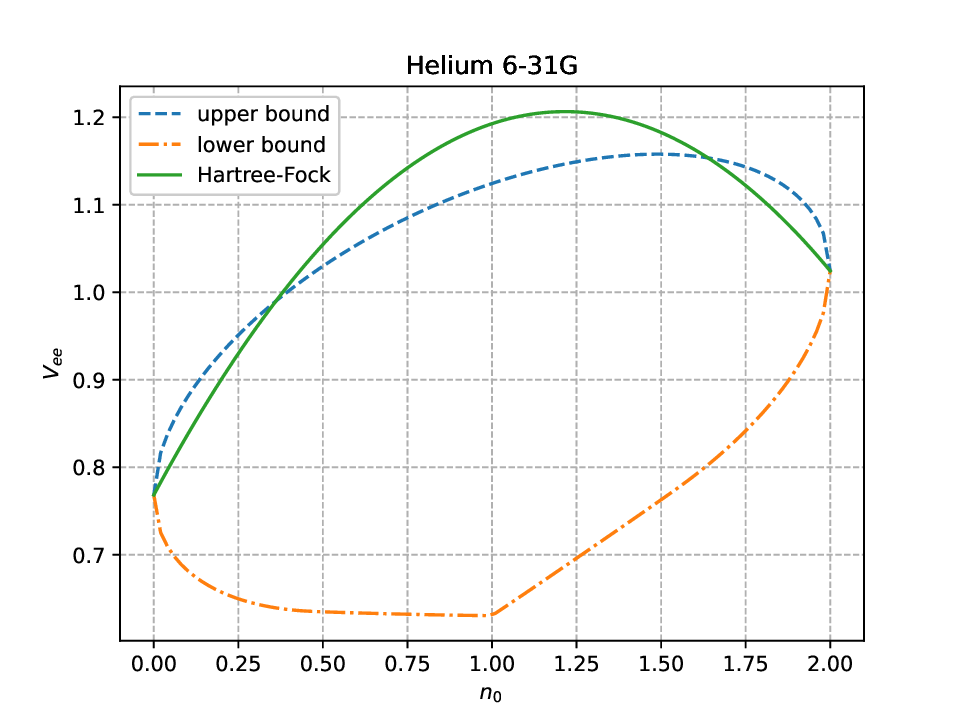}
    \caption{Upper and lower bound on the interparticle repulsion energy functional as a function of the natural orbital occupation number.\cite{note}}
    \label{fig:numercial_data}
\end{figure}

\section{Conclusion}
This paper aims to elucidate the $N$-representability problem in 1DMFT. We present necessary and sufficient conditions for the $N$-representability of the universal functional. By establishing rigorous representability conditions, our results lay the groundwork for developing more physically grounded and systematically improvable 1DMFT functionals.

While $N$-representability is necessary for the exact functional, it does not, by itself, ensure high accuracy in approximate functionals. Furthermore, verifying the conditions of both statements requires knowledge of the exact ground-state energy, rendering their evaluation an NP-hard problem except for small model systems and special cases\cite{Schuch2009}. Notably, the Hartree-Fock functional satisfies these conditions and thus provides an upper bound to the true ground-state energy \cite{LiebHartreeFock} for Coulomb repulsion, illustrating a rare tractable case. This is not true for Coulomb attraction, where the Hartree-Fock functional is not $N$-representable. It would be interesting to find a generalization of Lieb's argument from \cite{LiebHartreeFock} to attractive interactions, as it might give us a deeper understanding of approximations to the universal functional, leading to new functional development.

% Indeed, the only functional we are aware of that has been proven to be ensemble-$N$-representable for arbitrary $\gamma$ is the Hartree-Fock functional.\cite{liebVariationalPrincipleManyFermion1981} 

\section{Acknowledgements} The authors acknowledge support from the Canada Research Chairs (CRC-2022-00196), NSERC (Discovery RGPIN/06707-2024 and Alliance ALLRP/592521-2023), and the Digital Research Alliance of Canada. Moreover, we would like to thank Dr. Julia Liebert for many stimulating discussions.

% \printbibliography[title={Project- and subject-related list of publications}, heading=bibnumbered]

\bibliographystyle{apsrev4-1}
\bibliography{library}

@article{AyersNrepfDFTFs,
  title = {Necessary and sufficient conditions for the $N$-representability of density functionals},
  author = {Ayers, Paul W. and Liu, Shubin},
  journal = {Phys. Rev. A},
  volume = {75},
  issue = {2},
  pages = {022514},
  numpages = {12},
  year = {2007},
  month = {Feb},
  publisher = {American Physical Society},
  doi = {10.1103/PhysRevA.75.022514},
}

@article{PaulGolden,
    author = {Ayers, Paul W. and Golden, Sidney and Levy, Mel},
    title = {Generalizations of the Hohenberg-Kohn theorem: I. Legendre Transform Constructions of Variational Principles for Density Matrices and Electron Distribution Functions},
    journal = {J. Chem. Phys.},
    volume = {124},
    number = {5},
    pages = {054101},
    year = {2006},
    month = {02},
    issn = {0021-9606},
    doi = {10.1063/1.2006087},
}

@Article{appl_rep_cond_2,
author ="Teale, Andrew M. and Helgaker, Trygve and Savin, Andreas and Adamo, Carlo and Aradi, Bálint and Arbuznikov, Alexei V. and Ayers, Paul W. and Baerends, Evert Jan and Barone, Vincenzo and Calaminici, Patrizia and Cancès, Eric and Carter, Emily A. and Chattaraj, Pratim Kumar and Chermette, Henry and Ciofini, Ilaria and Crawford, T. Daniel and De Proft, Frank and Dobson, John F. and Draxl, Claudia and Frauenheim, Thomas and Fromager, Emmanuel and Fuentealba, Patricio and Gagliardi, Laura and Galli, Giulia and Gao, Jiali and Geerlings, Paul and Gidopoulos, Nikitas and Gill, Peter M. W. and Gori-Giorgi, Paola and Görling, Andreas and Gould, Tim and Grimme, Stefan and Gritsenko, Oleg and Jensen, Hans Jørgen Aagaard and Johnson, Erin R. and Jones, Robert O. and Kaupp, Martin and Köster, Andreas M. and Kronik, Leeor and Krylov, Anna I. and Kvaal, Simen and Laestadius, Andre and Levy, Mel and Lewin, Mathieu and Liu, Shubin and Loos, Pierre-François and Maitra, Neepa T. and Neese, Frank and Perdew, John P. and Pernal, Katarzyna and Pernot, Pascal and Piecuch, Piotr and Rebolini, Elisa and Reining, Lucia and Romaniello, Pina and Ruzsinszky, Adrienn and Salahub, Dennis R. and Scheffler, Matthias and Schwerdtfeger, Peter and Staroverov, Viktor N. and Sun, Jianwei and Tellgren, Erik and Tozer, David J. and Trickey, Samuel B. and Ullrich, Carsten A. and Vela, Alberto and Vignale, Giovanni and Wesolowski, Tomasz A. and Xu, Xin and Yang, Weitao",
title  ="DFT exchange: sharing perspectives on the workhorse of quantum chemistry and materials science",
journal  ="Phys. Chem. Chem. Phys.",
year  ="2022",
volume  ="24",
issue  ="47",
pages  ="28700-28781",
publisher  ="The Royal Society of Chemistry",
doi  ="10.1039/D2CP02827A"}

@article{appl_rep_cond_1,
    author = {Perdew, John P. and Ruzsinszky, Adrienn and Tao, Jianmin and Staroverov, Viktor N. and Scuseria, Gustavo E. and Csonka, Gábor I.},
    title = {Prescription for the design and selection of density functional approximations: More constraint satisfaction with fewer fits},
    journal = {J. Chem. Phys.},
    volume = {123},
    number = {6},
    pages = {062201},
    year = {2005},
    month = {08},
    issn = {0021-9606},
    doi = {10.1063/1.1904565},
}

@article{LiebOnDFT,
author = {Lieb, Elliott H.},
title = {Density functionals for coulomb systems},
journal = {Int. J. Quantum Chem.},
volume = {24},
number = {3},
pages = {243-277},
doi = {https://doi.org/10.1002/qua.560240302},
year = {1983}
}

@book{FuncAna,
  title={Functional Analysis: An Introduction},
  author={Larsen, R.},
  isbn={9780608169361},
  lccn={72090375},
  series={Pure Appl. Math. Ser.},
  year={1973},
  publisher={Marcel Dekker, Incorporated}
}

@article{GilbertExtension,
  title = {Hohenberg-Kohn theorem for nonlocal external potentials},
  author = {Gilbert, T. L.},
  journal = {Phys. Rev. B},
  volume = {12},
  issue = {6},
  pages = {2111--2120},
  numpages = {0},
  year = {1975},
  month = {Sep},
  publisher = {American Physical Society},
  doi = {10.1103/PhysRevB.12.2111},
}

@article{HarrimanAndHerbert,
    author = {Herbert, John M. and Harriman, John E.},
    title = {N-representability and variational stability in natural orbital functional theory},
    journal = {J. Chem. Phys.},
    volume = {118},
    number = {24},
    pages = {10835-10846},
    year = {2003},
    month = {06},
    issn = {0021-9606},
    doi = {10.1063/1.1574787},
}

@article{LevyValoneValone,
    author = {Valone, Steven M.},
    title = {Consequences of extending 1‐matrix energy functionals from pure–state representable to all ensemble representable 1 matrices},
    journal = {J. Chem. Phys.},
    volume = {73},
    number = {3},
    pages = {1344-1349},
    year = {1980},
    month = {08},
    issn = {0021-9606},
    doi = {10.1063/1.440249},
}

@article{LevyValoneLevy,
author = {Mel Levy },
title = {Universal variational functionals of electron densities, first-order density matrices, and natural spin-orbitals and solution of the <i>v</i>-representability problem},
journal = {Proc. Natl. Acad. Sci.},
volume = {76},
number = {12},
pages = {6062-6065},
year = {1979},
doi = {10.1073/pnas.76.12.6062},
URL = {https://www.pnas.org/doi/abs/10.1073/pnas.76.12.6062},
eprint = {https://www.pnas.org/doi/pdf/10.1073/pnas.76.12.6062}}

@article{ColemanVariationandNrep,
  title = {Structure of Fermion Density Matrices},
  author = {Coleman, A. J.},
  journal = {Rev. Mod. Phys.},
  volume = {35},
  issue = {3},
  pages = {668--686},
  numpages = {0},
  year = {1963},
  month = {Jul},
  publisher = {American Physical Society},
  doi = {10.1103/RevModPhys.35.668},
}

@article{Loewdin_I,
  title = {Quantum Theory of Many-Particle Systems. I. Physical Interpretations by Means of Density Matrices, Natural Spin-Orbitals, and Convergence Problems in the Method of Configurational Interaction},
  author = {L\"owdin, Per-Olov},
  journal = {Phys. Rev.},
  volume = {97},
  issue = {6},
  pages = {1474--1489},
  numpages = {0},
  year = {1955},
  month = {Mar},
  publisher = {American Physical Society},
  doi = {10.1103/PhysRev.97.1474},
}

@article{Loewdin_II,
  title = {Quantum Theory of Many-Particle Systems. II. Study of the Ordinary Hartree-Fock Approximation},
  author = {L\"owdin, Per-Olov},
  journal = {Phys. Rev.},
  volume = {97},
  issue = {6},
  pages = {1490--1508},
  numpages = {0},
  year = {1955},
  month = {Mar},
  publisher = {American Physical Society},
  doi = {10.1103/PhysRev.97.1490},
}

@article{LiebHartreeFock,
  title = {Variational Principle for Many-Fermion Systems},
  author = {Lieb, Elliot H.},
  journal = {Phys. Rev. Lett.},
  volume = {46},
  issue = {7},
  pages = {457--459},
  numpages = {0},
  year = {1981},
  month = {Feb},
  publisher = {American Physical Society},
  doi = {10.1103/PhysRevLett.46.457},
}

@article{Baerends_2005,
author = {Gritsenko, Oleg and Pernal, Katarzyna and Baerends, Evert},
year = {2005},
month = {06},
pages = {204102},
title = {An improved density matrix functional by physically motivated repulsive corrections},
volume = {122},
journal = {J. Chem. Phys.},
doi = {10.1063/1.1906203}
}

@article{MULLER_1984,
title = {Explicit approximate relation between reduced two- and one-particle density matrices},
journal = {Physics Letters A},
volume = {105},
number = {9},
pages = {446-452},
year = {1984},
issn = {0375-9601},
doi = {https://doi.org/10.1016/0375-9601(84)91034-X},
author = {A.M.K. Müller}
}

@article{Sharma2008,
  title = {Reduced density matrix functional for many-electron systems},
  author = {Sharma, S. and Dewhurst, J. K. and Lathiotakis, N. N. and Gross, E. K. U.},
  journal = {Phys. Rev. B},
  volume = {78},
  issue = {20},
  pages = {201103},
  numpages = {4},
  year = {2008},
  month = {Nov},
  publisher = {American Physical Society},
  doi = {10.1103/PhysRevB.78.201103},
}

@article{Piris_2006,
author = {Piris, Mario},
title = {A new approach for the two-electron cumulant in natural orbital functional theory},
journal = {Int. J. Quantum Chem.},
volume = {106},
number = {5},
pages = {1093-1104},
doi = {https://doi.org/10.1002/qua.20858},
year = {2006}
}

@incollection{Piris_2024,
title = {Chapter Two - Advances in approximate natural orbital functionals: From historical perspectives to contemporary developments},
editor = {Ramon A. Miranda Quintana and John F. Stanton},
series = {Advances in Quantum Chemistry},
publisher = {Academic Press},
volume = {90},
pages = {15-66},
year = {2024},
booktitle = {Novel Treatments of Strong Correlations},
issn = {0065-3276},
doi = {https://doi.org/10.1016/bs.aiq.2024.04.002},
author = {Mario Piris}
}

@article{Piris_2024b,
    author = {Mitxelena, Ion and Piris, Mario},
    title = "{Assessing the global natural orbital functional approximation on model systems with strong correlation}",
    journal = {J. Chem. Phys.},
    volume = {160},
    number = {20},
    pages = {204106},
    year = {2024},
    month = {05},
    issn = {0021-9606},
    doi = {10.1063/5.0207325},
}

@article{PNOF5,
author = {Piris, Mario},
title = {A new approach for the two-electron cumulant in natural orbital functional theory},
journal = {Int. J. Quantum Chem.},
volume = {106},
number = {5},
pages = {1093-1104},
doi = {https://doi.org/10.1002/qua.20858},
year = {2006}
}

@article{PNOF6,
    author = {Piris, M.},
    title = {Interacting pairs in natural orbital functional theory},
    journal = {J. Chem. Phys.},
    volume = {141},
    number = {4},
    pages = {044107},
    year = {2014},
    month = {07},
    issn = {0021-9606},
    doi = {10.1063/1.4890653},
}

@article{PNOF7,
  title = {Global Method for Electron Correlation},
  author = {Piris, Mario},
  journal = {Phys. Rev. Lett.},
  volume = {119},
  issue = {6},
  pages = {063002},
  numpages = {5},
  year = {2017},
  month = {Aug},
  publisher = {American Physical Society},
  doi = {10.1103/PhysRevLett.119.063002},
}

@article{GNOF,
  title = {Global Natural Orbital Functional: Towards the Complete Description of the Electron Correlation},
  author = {Piris, Mario},
  journal = {Phys. Rev. Lett.},
  volume = {127},
  issue = {23},
  pages = {233001},
  numpages = {6},
  year = {2021},
  month = {Dec},
  publisher = {American Physical Society},
  doi = {10.1103/PhysRevLett.127.233001},
}

@Article{PirisConcurrent,
author ="Piris, Mario",
title  ="Exploring the potential of natural orbital functionals",
journal  ="Chem. Sci.",
year  ="2024",
volume  ="15",
issue  ="42",
pages  ="17284-17291",
publisher  ="The Royal Society of Chemistry",
doi  ="10.1039/D4SC05810K"}

@article{Cohen2008,
author = {Aron J. Cohen  and Paula Mori-Sánchez  and Weitao Yang },
title = {Insights into Current Limitations of Density Functional Theory},
journal = {Science},
volume = {321},
number = {5890},
pages = {792-794},
year = {2008},
doi = {10.1126/science.1158722},
URL = {https://www.science.org/doi/abs/10.1126/science.1158722},
eprint = {https://www.science.org/doi/pdf/10.1126/science.1158722}}

@article{Donnelly_1979,
    author = {Donnelly, Robert A.},
    title = "{On a fundamental difference between energy functionals based on first‐ and on second‐order density matrices}",
    journal = {J. Chem. Phys.},
    volume = {71},
    number = {7},
    pages = {2874-2879},
    year = {1979},
    month = {10},
    issn = {0021-9606},
    doi = {10.1063/1.438678},
    comment_eprint = {https://pubs.aip.org/aip/jcp/article-pdf/71/7/2874/18919814/2874\_1\_online.pdf},
}

@Article{PirisBenchmark,
author ="Rodríguez-Mayorga, Mauricio and Ramos-Cordoba, Eloy and Via-Nadal, Mireia and Piris, Mario and Matito, Eduard",
title  ="Comprehensive benchmarking of density matrix functional approximations",
journal  ="Phys. Chem. Chem. Phys.",
year  ="2017",
volume  ="19",
issue  ="35",
pages  ="24029-24041",
publisher  ="The Royal Society of Chemistry",
doi  ="10.1039/C7CP03349D"}

@article{MachineLearned1RDMFT,
    author = {Franco, Lizeth and Bonfil-Rivera, Iván A. and Huan Lew-Yee, Juan Felipe and Piris, Mario and M. del Campo, Jorge and Vargas-Hernández, Rodrigo A.},
    title = {Softmax parameterization of the occupation numbers for natural orbital functionals based on electron pairing approaches},
    journal = {J. Chem. Phys.},
    volume = {160},
    number = {24},
    pages = {244107},
    year = {2024},
    month = {06},
    issn = {0021-9606},
    doi = {10.1063/5.0213719},
}

@article{SavinAC,
    author = {Colonna, François and Savin, Andreas},
    title = {Correlation energies for some two- and four-electron systems along the adiabatic connection in density functional theory},
    journal = {J. Chem. Phys.},
    volume = {110},
    number = {6},
    pages = {2828-2835},
    year = {1999},
    month = {02},
    issn = {0021-9606},
    doi = {10.1063/1.478234},
}

@article{Savin_2003,
author = {Savin, A. and Colonna, F. and Pollet, R.},
title = {Adiabatic connection approach to density functional theory of electronic systems},
journal = {Int. J. Quantum Chem.},
volume = {93},
number = {3},
pages = {166-190},
doi = {https://doi.org/10.1002/qua.10551},
year = {2003}
}

@article{Goerling_1996,
  title = {Generalized Kohn-Sham schemes and the band-gap problem},
  author = {Seidl, A. and G\"orling, A. and Vogl, P. and Majewski, J. A. and Levy, M.},
  journal = {Phys. Rev. B},
  volume = {53},
  issue = {7},
  pages = {3764--3774},
  numpages = {0},
  year = {1996},
  month = {Feb},
  publisher = {American Physical Society},
  doi = {10.1103/PhysRevB.53.3764},
}

@article{SCAN,
  title = {Strongly Constrained and Appropriately Normed Semilocal Density Functional},
  author = {Sun, Jianwei and Ruzsinszky, Adrienn and Perdew, John P.},
  journal = {Phys. Rev. Lett.},
  volume = {115},
  issue = {3},
  pages = {036402},
  numpages = {6},
  year = {2015},
  month = {Jul},
  publisher = {American Physical Society},
  doi = {10.1103/PhysRevLett.115.036402},
}

@article{Becke,
    author = {Becke, Axel D.},
    title = {Density‐functional thermochemistry. III. The role of exact exchange},
    journal = {J. Chem. Phys.},
    volume = {98},
    number = {7},
    pages = {5648-5652},
    year = {1993},
    month = {04},
    issn = {0021-9606},
    doi = {10.1063/1.464913},
}

@article{PBE,
  title = {Generalized Gradient Approximation Made Simple},
  author = {Perdew, John P. and Burke, Kieron and Ernzerhof, Matthias},
  journal = {Phys. Rev. Lett.},
  volume = {77},
  issue = {18},
  pages = {3865--3868},
  numpages = {0},
  year = {1996},
  month = {Oct},
  publisher = {American Physical Society},
  doi = {10.1103/PhysRevLett.77.3865},
}

@article{scaledSigma,
    author = {Erhard, Jannis and Fauser, Steffen and Trushin, Egor and Görling, Andreas},
    title = {Scaled σ-functionals for the Kohn–Sham correlation energy with scaling functions from the homogeneous electron gas},
    journal = {J. Chem. Phys.},
    volume = {157},
    number = {11},
    pages = {114105},
    year = {2022},
    month = {09},
    issn = {0021-9606},
    doi = {10.1063/5.0101641},
}

@article{PSA,
  title = {Power Series Approximation for the Correlation Kernel Leading to Kohn-Sham Methods Combining Accuracy, Computational Efficiency, and General Applicability},
  author = {Erhard, Jannis and Bleiziffer, Patrick and G\"orling, Andreas},
  journal = {Phys. Rev. Lett.},
  volume = {117},
  issue = {14},
  pages = {143002},
  numpages = {6},
  year = {2016},
  month = {Sep},
  publisher = {American Physical Society},
  doi = {10.1103/PhysRevLett.117.143002},
}

@article{Scuseria_2008,
  title = {Density functional with full exact exchange, balanced nonlocality of correlation,  and constraint satisfaction},
  author = {Perdew, John P. and Staroverov, Viktor N. and Tao, Jianmin and Scuseria, Gustavo E.},
  journal = {Phys. Rev. A},
  volume = {78},
  issue = {5},
  pages = {052513},
  numpages = {13},
  year = {2008},
  month = {Nov},
  publisher = {American Physical Society},
  doi = {10.1103/PhysRevA.78.052513},
}

@article{Schuch2009,
  author    = {Norbert Schuch and Frank Verstraete},
  title     = {Computational complexity of interacting electrons and fundamental limitations of density functional theory},
  journal   = {Nat. Phys.},
  volume    = {5},
  number    = {10},
  pages     = {732--735},
  year      = {2009},
  month     = {oct},
  issn      = {1745-2481},
  doi       = {10.1038/nphys1370}
}

@book{mezeyManyelectronDensitiesReduced2000,
  title = {Many-Electron Densities and Reduced Density Matrices},
  editor = {Mezey, P. G. and Cioslowski, J.},
  year = {2000},
  series = {Mathematical and {{Computational Chemistry}}},
  publisher = {Springer},
  address = {New York}
}

@article{tealeAccurateCalculationModeling2010,
  title = {Accurate Calculation and Modeling of the Adiabatic Connection in Density Functional Theory},
  author = {Teale, A. M. and Coriani, S. and Helgaker, T.},
  year = {2010},
  month = apr,
  journal = {Journal of Chemical Physics},
  volume = {132},
  pages = {164115},
  issn = {0021-9606},
  doi = {10.1063/1.3380834}
}

@article{tealeCalculationAdiabaticconnectionCurves2009,
  title = {The Calculation of Adiabatic-Connection Curves from Full Configuration-Interaction Densities: {{Two-electron}} Systems},
  author = {Teale, A. M. and Coriani, S. and Helgaker, T.},
  year = {2009},
  month = mar,
  journal = {Journal of Chemical Physics},
  volume = {130},
  pages = {104111},
  issn = {0021-9606},
  doi = {10.1063/1.3082285}
}

@incollection{lewinUniversalFunctionalsDensity2023,
  title = {Universal {{Functionals}} in {{Density Functional Theory}}},
  booktitle = {Density {{Functional Theory}}: {{Modeling}}, {{Mathematical Analysis}}, {{Computational Methods}}, and {{Applications}}},
  author = {Lewin, Mathieu and Lieb, Elliott H. and Seiringer, Robert},
  editor = {Canc{\`e}s, Eric and Friesecke, Gero},
  year = {2023},
  pages = {115--182},
  publisher = {Springer International Publishing},
  address = {Cham},
  doi = {10.1007/978-3-031-22340-2_3},
  urldate = {2025-06-17},
  isbn = {978-3-031-22340-2},
  langid = {english}
}

@incollection{tealeLiebVariationPrinciple2022,
  title = {Lieb Variation Principle in Density-Functional Theory},
  booktitle = {The {{Physics}} and {{Mathematics}} of {{Elliott Lieb}}},
  author = {Teale, Andrew M. and Helgaker, Trygve},
  year = {2022},
  volume = {1},
  pages = {527--559},
  publisher = {EMS Press},
  urldate = {2025-06-17},
  langid = {english}
}

@article{ayersAxiomaticFormulationsHohenbergKohn2006,
  title = {Axiomatic Formulations of the {{Hohenberg-Kohn}} Functional},
  author = {Ayers, Paul},
  year = {2006},
  month = jan,
  journal = {Phys. Rev. A},
  volume = {73},
  number = {1},
}

@article{ayersGeneralizedDensityfunctionalTheory2005,
  title = {Generalized Density-Functional Theory: {{Conquering theN-representability}} Problem with Exact Functionals for the Electron Pair Density and the Second-Order Reduced Density Matrix},
  shorttitle = {Generalized Density-Functional Theory},
  author = {Ayers, Paul W. and Levy, Mel},
  year = {2005},
  month = sep,
  journal = {Journal of Chemical Sciences},
  volume = {117},
  number = {5},
  pages = {507--514},
  issn = {0973-7103},
  doi = {10.1007/BF02708356},
  urldate = {2025-02-09},
  langid = {english}
}

@article{mori-sanchezManyelectronSelfinteractionError2006,
  title = {Many-Electron Self-Interaction Error in Approximate Density Functionals},
  author = {{Mori-Sanchez}, P. and Cohen, A. J. and Yang, W. T.},
  year = {2006},
  journal = {Journal of Chemical Physics},
  volume = {125},
  pages = {201102}
}

@article{ludenaFunctionalNrepresentability2matrix2013,
  title = {Functional {{N-representability}} in 2-Matrix, 1-Matrix, and Density Functional Theories},
  author = {Ludena, E. V. and Torres, F. J. and Costa, C.},
  year = {2013},
  journal = {J. Mod. Phys.},
  volume = {4},
  pages = {29318}
}

@article{ludenaNrepresentabilityUniversalityFrho2009,
  title = {On the {{N-representability}} and Universality of {{F}}[Rho] in the {{Hohenberg-Kohn-Sham}} Version of Density Functional Theory},
  author = {Ludena, E. V. and Illas, F. and {Ramirez-Solis}, A.},
  year = {2009},
  journal = {Condens. Matter Theor., Vol 23},
  pages = {354--366},
  issn = {978-981-283-661-8},
  langid = {english}
}

@article{kryachkoFormulationNRepresentableUpsilonRepresentable1991,
  title = {Formulation of {{N-Representable}} and {{Upsilon-Representable Density-Functional Theory}} .1. {{Ground-States}}},
  author = {Kryachko, E. S. and Ludena, E. V.},
  year = {1991},
  journal = {Physical Review A},
  volume = {43},
  pages = {2179--2193}
}

@article{ludenaHohenbergKohnShamVersionDFT2004,
  title = {Is the {{Hohenberg-Kohn-Sham}} Version of {{DFT}} a Semi-Empirical Theory?},
  author = {Ludena, E. V.},
  year = {2004},
  journal = {Journal of Molecular Structure-Theochem},
  volume = {709},
  pages = {25--29}
}

@article{cuevas-saavedraSymmetricNonLocal2012,
  title = {Symmetric {{Non}} Local {{Weighted Density Approximations}} from the {{Exchange-Correlation Hole}} of the {{Uniform Electron Gas}}},
  author = {{Cuevas-Saavedra}, R. and Chakraborty, D. and Rabi, S. and Cardenas, C. and Ayers, P. W.},
  year = {2012},
  month = nov,
  journal = {J. Chem. Theory Comput.},
  volume = {8},
  number = {11},
  pages = {4081--4093},
  issn = {1549-9618},
  doi = {10.1021/ct300325t},
  langid = {english}
}

@article{cioslowskiVariationalDensityMatrix2002,
  title = {Variational Density Matrix Functional Theory Calculations with the Lowest-Order {{Yasuda}} Functional},
  author = {Cioslowski, J. and Pernal, K.},
  year = {2002},
  journal = {Journal of Chemical Physics},
  volume = {117},
  pages = {67--71}
}

@article{feliciangeliNoncommutativeEntropicOptimal2023,
  title = {A Non-Commutative Entropic Optimal Transport Approach to Quantum Composite Systems at Positive Temperature},
  author = {Feliciangeli, Dario and Gerolin, Augusto and Portinale, Lorenzo},
  year = {2023},
  month = aug,
  journal = {J. Funct. Anal.},
  volume = {285},
  number = {4},
  pages = {109963},
  issn = {0022-1236},
  doi = {10.1016/j.jfa.2023.109963},
  urldate = {2025-06-23}
}

@misc{fredheimReducedDensityMatrix2025,
Author = {Håkon Richard Fredheim and Simen Kvaal},
Title = {Reduced Density Matrix Functional Theory And A Reduced Formulation Of Density Functional Theory},
Year = {2025},
Eprint = {arXiv:2510.12242},
}

@misc{note,
  author = {Erhad, Jannis and Ayers, Paul},
  title = {Data to support figure 1. Version 1.0},
  note = {https://github.com/theochem/research-data},
  year = {2025}
}

@article{Maziotti2023,
  title = {Universal Generalization of Density Functional Theory for Static Correlation},
  author = {Gibney, Daniel and Boyn, Jan-Niklas and Mazziotti, David A.},
  journal = {Phys. Rev. Lett.},
  volume = {131},
  issue = {24},
  pages = {243003},
  numpages = {8},
  year = {2023},
  month = {Dec},
  publisher = {American Physical Society},
  doi = {10.1103/PhysRevLett.131.243003},
}

@article{gebauerWellscalingNaturalOrbital2016,
  title = {A Well-Scaling Natural Orbital Theory},
  author = {Gebauer, R. and Cohen, M. H. and Car, R.},
  year = 2016,
  month = nov,
  journal = {Proceedings of the National Academy of Sciences of the United States of America},
  volume = {113},
  number = {46},
  pages = {12913--12918},
  issn = {0027-8424},
  doi = {10.1073/pnas.1615729113}
}

@article{senjeanReducedDensityMatrix2022,
  title = {Reduced Density Matrix Functional Theory from an Ab Initio Seniority-Zero Wave Function: {{Exact}} and Approximate Formulations along Adiabatic Connection Paths},
  shorttitle = {Reduced Density Matrix Functional Theory from an Ab Initio Seniority-Zero Wave Function},
  author = {Senjean, Bruno and Yalouz, Saad and Nakatani, Naoki and Fromager, Emmanuel},
  year = 2022,
  month = sep,
  journal = {Physical Review A},
  volume = {106},
  number = {3},
  pages = {032203},
  publisher = {American Physical Society},
  doi = {10.1103/PhysRevA.106.032203},
  urldate = {2024-05-01}
}

@article{ayersDensityfunctionalTheoryAdditional2009b,
  title = {Density-Functional Theory with Additional Basic Variables: {{Extended Legendre}} Transform},
  author = {Ayers, P. W. and Fuentealba, P.},
  year = 2009,
  month = sep,
  journal = {Physical Review A},
  volume = {80},
  pages = {032510},
  issn = {1050-2947},
  doi = {032510 10.1103/PhysRevA.80.032510},
  langid = {english}
}

% \printbibliography

\end{document}